\documentclass[12pt]{article}
\usepackage[margin=1in]{geometry}

\usepackage{amsmath}
\usepackage{amssymb}
\usepackage{amsthm}
\usepackage{xspace}
\usepackage[normalem]{ulem}
\usepackage{graphicx}
\usepackage{epsfig}
\usepackage{ifpdf}
\usepackage{url,hyperref}
\usepackage{latexsym}
\usepackage[mathscr]{euscript}
\usepackage{xspace}
\usepackage{color}
\usepackage{makeidx}
\usepackage{wrapfig,color}
\usepackage{stackrel}
\usepackage{algorithm}
\usepackage{authblk}
\usepackage[noend]{algpseudocode}
\usepackage{todonotes}
\usepackage{tikz-cd}

\usepackage[all]{xy}
\usepackage{framed}
\usepackage{pb-diagram}

\long\def\remove#1{}

\DeclareMathOperator{\dgm}{Dgm}

\DeclareMathOperator{\xiuv}{\xi_{\{u,v\}}}

\DeclareMathOperator{\maxv}{MaxV}
\DeclareMathOperator{\minv}{MinV}
\DeclareMathOperator{\maxt}{MaxT}
\DeclareMathOperator{\mint}{MinT}

\DeclareMathOperator{\lk}{Lk}
\DeclareMathOperator{\cl}{Cl}
\DeclareMathOperator{\st}{Star}

\newtheorem{theorem}{\sffamily Theorem}
\newtheorem{lemma}[theorem]{\sffamily Lemma}
\newtheorem{corollary}[theorem]{\sffamily Corollary}

\newtheorem{definition}[theorem]{\sffamily Definition}
\newtheorem{remark}[theorem]{\sffamily Remark}

\newif\ifpaper
\papertrue

\title{Filtration Simplification for Persistent Homology via Edge Contraction}

\begin{document}

\author{Tamal K. Dey\thanks{dey.8@osu.edu}}

\author{Ryan Slechta\thanks{slechta.3@osu.edu}}

\affil{Department of Computer Science and Engineering,
The Ohio State University, Columbus, OH 43210, USA.}

\maketitle

\setcounter{page}{1}
\begin{abstract}
Persistent homology is a popular data analysis technique that is used to capture the changing topology of a filtration associated with some simplicial complex $K$. These topological changes are summarized in persistence diagrams. We propose two contraction operators which when applied to $K$ and its associated filtration, bound the perturbation in the persistence diagrams. The first assumes that the underlying space of $K$ is a $2$-manifold and ensures that simplices are paired with the same simplices in the contracted complex as they are in the original. The second is for arbitrary $d$-complexes, and bounds the bottleneck distance between the initial and contracted $p$-dimensional persistence diagrams. This is accomplished by defining interleaving maps between persistence modules which arise from chain maps defined over the filtrations. In addition, we show how the second operator can efficiently compose across multiple contractions. We conclude with experiments demonstrating the second operator's utility on manifolds.
\end{abstract}

\section{Introduction}
\label{sec:intro}

Edge contraction is a fundamental operation which has been famously explored by the graphics and computational geometry communities when developing tools for mesh simplification \cite{linkcondition,Hoppe,polygsimplification} and by mathematicians when developing graph minor theory \cite{ROBERTSON1}. However, comparatively little work has been done to incorporate edge contraction as a tool for topological data analysis. Edge contraction has been used to compute persistent homology for simplicial maps \cite{DFW14} and to simplify discrete Morse vector fields \cite{DiscreteMorseCommute,MarylandContraction}, but no work has been done to develop a \textit{persistence-aware} contraction operator.
\textit{Persistent homology} is based on the observation that adding a simplex to a simplicial complex either creates or destroys a homology class \cite{EdelsPers}. Hence, the lifetime, or \textit{persistence}, of a class through a filtered simplicial complex can be defined as the difference in the birth time and death time of the class. In addition, this permits a pairing of simplices, where $\sigma$ is paired with $\tau$ if $\tau$ destroys the homology class created by $\sigma$. A summary of the births and deaths of homology classes is given in a \textit{persistence diagram}. We give further details in Section \ref{sec:prelim}.
 
In this paper, we aim to develop contraction operators which when applied to a filtered simplicial complex, simplify the cell structure while also controlling perturbations in the persistence diagrams associated with the complexes. We develop two such operators: one for $2$-manifolds which maintains the same pairing in the contracted complex as the original, and one for arbitrary $d$-complexes which bounds the bottleneck distance between the persistence diagrams of the original and contracted filtrations. In addition, we show how our operator for $d$-complexes composes with itself to bound perturbation across multiple contractions. We provide an implementation of the operator which controls bottleneck distance and demonstrate its utility on manifolds. 
\section{Preliminaries}
\label{sec:prelim}
Throughout this paper, we will use $K$ to refer to a finite simplicial complex of arbitrary dimension, unless otherwise specified. We assume that $K$ is equipped with a height function 
\begin{equation}
h\; : \; K \to \mathbb{R} 
\end{equation}
such that if $\sigma$ is a face of $\tau$, $h(\sigma) \leq h(\tau)$. This is equivalent to assuming that $K$ is \textit{filtered}. That is, $K$ is equipped with a sequence of subcomplexes $\{K_a\}_{a\in A}$ where $K_a \subset K_{a'}$ if $a < a'$, $|A|$ finite. In addition, for some $a\in A$, $K_a = K$. A filtration induces a height function on $K$ that respects the face poset, where the height of any particular simplex is the first index at which it occurs. Similarly, $h$ induces a filtration in a canonical way. 

The filtration $\{K_a\}_{a\in A}$ gives a natural partial order on the simplices of $K$. It induces a total order $\prec$ by giving precedence to lower dimensional simplices and arbitrarily breaking ties within each dimension. For simplices $\sigma,\tau \in K$, we write $\sigma < \tau$ if $\sigma$ is a face of $\tau$, and define $\sigma > \tau$ as expected. Similarly, we write $\sigma <_1 \tau$ if $\sigma$ is a facet of $\tau$. If $\sigma \neq \tau$ are of the same dimension and there exists a $\rho$ where $\rho <_1 \sigma$ and $\rho <_1 \tau$, then we say that $\sigma$ is \textit{incident} to $\tau$ (and vice-versa). 
\subsection{Edge Contraction}
For a filtered simplicial complex $K$, we model contracting edge $\{u,v\} \in K$ as a simplicial map 
\begin{equation}
    \xiuv \; : \; K \to \Delta
\end{equation}
where $\Delta$ is the maximal simplicial complex on the vertex set of $K$. We often denote $\xiuv(K)$ as $K'$. A height function $h' \; : \; K' \to \mathbb{R}$ is induced on $K'$ by defining $h'(\sigma) = \min\{ h(\tau) \; | \; \tau \in \xiuv^{-1}(\sigma)\}$. Equivalently, the filtration $\{\xiuv(K_a)\}_{a \in A}$ is induced on $K'$. The total order $\prec$ also induces a total order on $K'$ where if multiple simplices map to the same simplex, then the image takes the first position in the total order of its preimages. 
We will abuse notation and allow $\prec$ to refer to both the total order on $K$ and $K'$.

In this paper, when contracting $\{u,v\}$, we always assume that $u \prec v$. Following \cite{DFW14}, those simplices $\sigma$ for which $\{u,v\} < \sigma$ are called \textit{vanishing simplices}. If $\sigma$ is the face of some vanishing simplex and contains exactly one of $u$ or $v$ as a face, then $\sigma$ is a \textit{mirrored simplex}. 
\begin{remark}
Mirrored simplices come in pairs. If $\sigma$ is defined by the $n$ vertices $\{x_1, x_2, \ldots, x_{n-1}, u\}$, then the \textit{mirror} of $\sigma$, denoted $m(\sigma)$, is the simplex defined by $\{x_1,x_2,\ldots,x_{n-1},v\}$. 
\end{remark}
If $\sigma$ is either mirrored or vanishing, then $\sigma$ is a \textit{local} simplex. If a simplex $\sigma$ contains a mirrored simplex as a facet, but is not a local simplex, then $\sigma$ is said to be an \textit{adjacent simplex}. Figure \ref{fig:simpDiag} details the various types of simplices in a $2$-complex. 
\begin{figure}\centering
  \includegraphics[scale=.3]{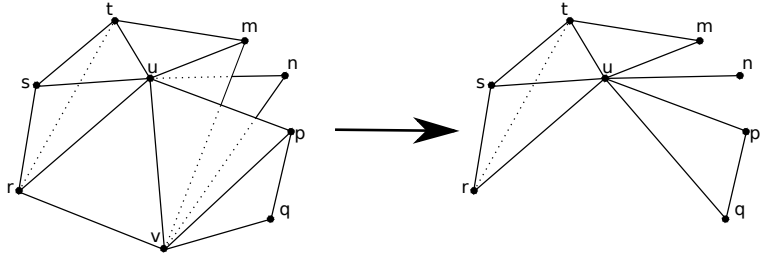}
  \caption{A simplicial $2$-complex before and after contracting edge $\{u,v\}$ where $u \prec v$. In the pre-contracted complex, edge $\{u,v\}$ and triangles $\{r,u,v\}$, $\{p,u,v\}$, $\{m,u,v\}$, and $\{n,u,v\}$ are vanishing. Vertex pair $(u,v)$ and edge pairs $(\{r,u\},\{r,v\})$, $(\{p,u\},\{p,v\})$, $(\{m,u\},\{m,v\})$, and $(\{n,u\},\{n,v\})$ are mirrored edges. Triangles $\{r,s,u\}$, $\{m,t,u\}$, $\{p,q,v\}$, $\{r,t,u\}$ and edges $\{s,u\}$, $\{t,u\}$, $\{p,v\}$ are adjacent simplices. All other simplices are generic nonlocal simplices. }
  \label{fig:simpDiag}
\end{figure}
For nonlocal simplices, $\xiuv$ is the identity. If $\sigma$ is a mirrored simplex, and $u < \sigma$, then $\xiuv(\sigma) = \xiuv(m(\sigma)) = \sigma$. Otherwise, $\xiuv(\sigma) = \xiuv(m(\sigma)) = m(\sigma)$. If $\tau$ is vanishing, then it has mirrored facet (or face of codimension one) $\sigma$. We let $\xiuv(\sigma)$ determine $\xiuv(\tau)$.

\subsection{Persistence Modules and Filtrations}

Let $\{K_a\}_{a\in A}$ denote some filtration for $K$. Note that if $a < b$, there is a natural map from the chains of $K_a$ to $K_b$. The chain maps define a map between the $p$-dimensional homology groups $H_p(K_a) \to H_p(K_b)$ (consult \cite{CompTop} for details).  These homology groups, together with all of the induced maps between them, give a \textit{persistence module}. Formally, we use the definition given by Chazal \textit{et. al.} \cite{Chazal}

\begin{definition}
Let $R$ be a commutative ring with unity, and $A$ a subset of $\mathbb{R}$. A persistence module $M_A$ is a family $\{F_\alpha\}_{\alpha \in A}$ of $R$-modules indexed by the elements of $A$, together with a family $\{f_\alpha^{\alpha'} \; : \; F_\alpha \to F_{\alpha'}\}_{\alpha \leq \alpha' \in A}$ of homomorphisms such that, $\forall \alpha \leq \alpha' \leq \alpha'' \in A, f_\alpha^{\alpha''} = f_{\alpha'}^{\alpha''} \circ f_\alpha^{\alpha'}$. 
\end{definition}

In this paper, we will only consider when $R$ is $\mathbb{Z}_2$. As mentioned earlier, $\{K_a\}_{a\in A}$ gives a persistence module

\begin{tikzcd}
\centering
M_A \; : \; H_p(K_{j_1}) \arrow[r,"f_{j_1,j_2}"]
& H_p(K_{j_2}) \arrow[r,"f_{j_2,j_3}"]
& H_p(K_{j_3}) \arrow[r,"f_{j_3,j_4}"]
& \cdots \arrow[r,"f_{j_{n-1},j_n}"]
& H_p(K_{j_n})
\end{tikzcd}

\noindent where additional maps are given by composition. Note that there is a second persistence module given by the filtration $\{\xiuv(K_a)\}_{a \in A}$. For convenience, we will define $K_a':=\xiuv(K_a)$. This gives us a module for the contracted complex:

\begin{tikzcd}
\centering
M'_A \; : \;H_k(K'_{j_1}) \arrow[r,"f'_{j_1,j_2}"]
& H_k(K'_{j_2}) \arrow[r,"f'_{j_2,j_3}"]
& H_k(K'_{j_3}) \arrow[r,"f'_{j_3,j_4}"]
& \cdots \arrow[r,"f'_{j_{n-1},j_n}"]
& H_k(K'_{j_n}).
\end{tikzcd}

\subsection{Persistence Diagrams}

A \textit{persistence diagram} captures the birth and deaths of homology classes in a corresponding persistence module. We let $\mu_p^{i,j}$ be the number of $p$-dimensional homology classes which are born in $K_i$ and die in $K_j$. This gives a formalization of a persistence diagram \cite{Cohen-Steiner}.
\begin{definition}
A persistence diagram $\dgm_p(f)$ of a filtration induced by $f$ is a multi-subset of the extended real plane, such that each point $(i,j), j > i$ has multiplicity $\mu_p^{i,j}$, and points $(i,i)$ have infinite multiplicity. 
\end{definition}
Persistence diagrams are often plotted as in Figure \ref{fig:persDiag}. Note that a persistence diagram can equivalently be thought to capture the changes in a persistence module, so for persistence module $M$ we often use the notation $\dgm_p(M)$. In addition, we can define a distance between persistence diagrams.
\begin{definition}
For persistence diagrams $\dgm_p(f), \dgm_p(g)$, let $B$ denote the set of all bijections $\gamma \; : \; \dgm_p(f) \to \dgm_p(g)$. The bottleneck distance $d_b$ is defined as
\[
    d_b(\dgm_p(f),\dgm_p(g)) = \inf_{\gamma \in B} \sup_{x \in \dgm(f)} d_{\infty}(x, \gamma(x))
\]
where $d_\infty$ denotes the infinity norm. 
\end{definition}

\begin{figure}\centering
  \includegraphics[scale=.21]{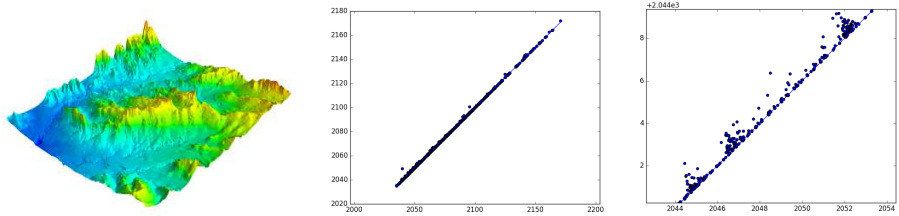}
  \caption{A triangulated terrain (left), the $0$-dimensional persistence diagram corresponding to the terrain (middle), and a closer view of the same diagram (right). The height function was extended to the entire complex such that the value at a simplex is the maximum height value of its constituent vertices.}
  \label{fig:persDiag}
\end{figure}

We encourage the reader to consult \cite{CompTop} for a more thorough treatment of persistence diagrams and bottleneck distance. Developing a contraction operator which bounds the perturbation in the bottleneck distance between the persistence diagrams of the original and contracted filtrations is one of the goals of this paper. However, this distance is quite cumbersome. Chazal et. al. showed that bottleneck distance could be directly related to the persistence module corresponding to a filtration \cite{Chazal}. This requires a notion of similarity between persistence modules. We use the definition given in \cite{Chazal}.
\begin{definition}
Two persistence modules $M_{\mathbb{R}}$ and $M'_{\mathbb{R}}$ are strongly $\epsilon$-interleaved if there exist two families of homomorphisms $\{\phi_\alpha \; : \; M_\alpha \to M'_{\alpha+\epsilon}\}_{\alpha \in \mathbb{R}}$ and $\{\psi_\alpha \; : \; M_\alpha \to M'_{\alpha + \epsilon}\}_{\alpha \in \mathbb{R}}$ such that the diagrams of Equation \ref{eqn:commute} commute $\forall \alpha \leq \alpha' \in \mathbb{R}$.
\end{definition}
\begin{equation}
\label{eqn:commute}
\begin{tikzcd}
\centering
M_{\alpha-\epsilon} \arrow[rrr] \arrow[dr] & & & M_{\alpha' + \epsilon} & & M_{\alpha + \epsilon} \arrow[r] & M_{\alpha' + \epsilon} \\
& M'_\alpha \arrow[r] & M'_{\alpha'} \arrow[ur] & & M'_{\alpha} \arrow[r] \arrow[ur] & M'_{\alpha'} \arrow[ur] & \\
& M_\alpha \arrow[r] & M_{\alpha'} \arrow[dr] & & M_{\alpha} \arrow[r] \arrow[dr] & M_{\alpha'} \arrow[dr] & \\
M'_{\alpha-\epsilon} \arrow[rrr] \arrow[ur] & & & M'_{\alpha' + \epsilon} & & M'_{\alpha + \epsilon} \arrow[r] & M'_{\alpha' + \epsilon}
\end{tikzcd}   
\end{equation}

Chazal \textit{et. al.} proved the following theorem. 
\begin{theorem}
\label{thm:bint}
Let $M_\mathbb{R}$ and $M'_\mathbb{R}$ be tame persistence modules. If $M_\mathbb{R}$ and $M'_\mathbb{R}$ are strongly $\epsilon$-interleaved, then $d_b(\dgm(M_\mathbb{R}), \dgm(M'_\mathbb{R})) \leq \epsilon$.
\end{theorem}
As $K$ is assumed to be finite, all persistence modules we consider will be tame. In addition, $\epsilon$-interleavings induce a pseudometric on the space of persistence modules called \textit{interleaving distance}. Lesnick showed in \cite{Lesnick2015} that the interleaving distance between two persistence modules is the same as the bottleneck distance between their persistence diagrams. This permits application of the triangle inequality when considering multiple contractions. 
\section{Preserving Pairings}
\label{sec:pair}

We now move to developing a contraction operator for $2$-manifolds such that if $\sigma$ is paired with $\tau$, and if $\sigma$ is nonlocal or a mirror that precedes its partner under $\prec$, then $\xiuv(\sigma)$ is paired with $\xiuv(\tau)$. 

\subsection{Pairings for Manifolds}

In this section, we assume that the underlying space (of the geometric realization) of $K$ is a $2$-manifold. In particular, we assume that $K$ is without boundary, but a slight modification works for manifolds with boundary. Hence, we can assume that each edge is the facet of exactly two triangles. Attali et. al. observed that the persistence pairings of such complexes can be computed in near-linear time in the number of edges\cite{Attali}. This is done by considering two graphs induced by the simplicial complex: the vertex graph $K_v$ and the triangle graph $K_t$. To avoid confusion, we refer to edges and vertices in the vertex and triangle graphs as \textit{arcs} and \textit{nodes}, respectively. The vertex graph is induced in the obvious way, and the triangle graph is its dual. Note that arcs in $K_t$ also correspond to edges in $K$. In both graphs, arc $e$ is weighted with value $h(e')$, where $e'$ is the edge corresponding to $e$. Persistence partners are computed by using Kruskal's algorithm to compute a minimum spanning tree on $K_v$ and a maximum spanning tree on $K_t$. When an arc $e$ is is introduced to a spanning forest on $K_v$, it connects two trees rooted at nodes $v_1,v_2$. Assuming the vertex corresponding to $v_1$ occurs prior to that corresponding to $v_2$ under $\prec$, we define $\minv(e) = v_1$ and $\maxv(e) = v_2$. Following the introduction of $e$, $\minv(e)$ becomes the root of the combined tree and the edge corresponding to $e$ is paired with the vertex corresponding to $\maxv(e)$. For arcs $e$ in $K_t$, $\maxt(e)$ and $\mint(e)$ are defined analogously, except upon the introduction of $e$, $\maxt(e)$ is the root of the new tree, and the edge corresponding to $e$ is paired with the triangle corresponding to $\mint(e)$.

Note that some simplices remain unpaired following this algorithm. We let $P(K)$ denote the set of \textit{pairs} of simplices that results from the aforementioned algorithm. We aim to develop a contraction operator such that 
\begin{equation}
    P(K') = \{(\xiuv(\sigma),\xiuv(\tau)) \; | \; (\sigma, \tau) \in P(K) \land \xiuv(\sigma) \neq \xiuv(\tau)\}.
\end{equation}

\subsection{A Persistence-Pair Preserving Condition}

We now present sufficient conditions for contracting an edge that maintains the persistence pairing. First, we assume that $e$ satisfies the \textit{link condition}, which ensures that the complex remains a $2$-manifold following contraction \cite{linkcondition}. For simplex $\sigma$, we define $\cl(\sigma) = \{\tau \; | \; \tau < \sigma\}$, $\st(\sigma) = \{\tau \; | \; \sigma < \tau\}$, and $\lk(\sigma) = \cl(\st(\sigma)) \setminus \st(\cl(\sigma))$. All of these operations extend to sets of simplices in the natural way. 
\begin{definition}[Link Condition\cite{linkcondition}]
An edge $e = \{u,v\}$ satisfies the link condition if $\lk(e) = \lk(u) \cap \lk(v)$. We say $e$ is \textit{contractible} if it satisfies the link condition.
\end{definition}
Requiring the link condition implies that there are two sets of mirrored edges relative to $\{u,v\}$. We label the two incident triangles to $\{u,v\}$ as $t_1$ and $t_2$, and their respective constituent mirrored edges as $e_{1},e_{1}'$ and $e_{2},e_{2}'$. We will assume without loss of generality that $e_{1} \prec e_{1}'$ and $e_{2} \prec e_{2}'$. 
\begin{definition}
An edge $e = \{u,v\}$ is admissible if $e$ satisfies the Link Condition, and $e$ is paired with $v$, $t_1$ is paired with $e_{1}'$, and $t_2$ is paired with $e_{2}'$.
\end{definition}
This definition, together with the pairing algorithm for $2$-manifolds, gives the following results.
\begin{theorem}
If $\{u,v\}$ is admissible, and $e \in K$, $e \neq \{u,v\}, e_{1}', e_{2}'$ is paired with vertex $r$, then $\xiuv(e)$ is paired with $\xiuv(r)$.
\end{theorem}
\begin{proof}
Aiming for a contradiction, we assume that $e$ is paired with $r$ but $\xiuv(e)$ is not paired with $\xiuv(r)$. In particular, we assume that $e$ is the first such edge under $\prec$ which satisfies this condition. Then upon introducing the arc corresponding to $\xiuv(e)$, there necessarily exist paths in $K'_v$ from the end nodes of the arc corresponding to $\xiuv(e)$ to $\xiuv(\maxv(e))$ and $\xiuv(\minv(e))$. In particular, both $\xiuv(\maxv(e))$ and $\xiuv(\minv(e))$ must be unpaired, as $e$ is the first edge such that $e$ is paired with $r$ but $\xiuv(e)$ is not paired with $\xiuv(r)$. Hence, this means that $\xiuv(\maxv(e)) \prec \xiuv(\minv(e))$. But this means that $\maxv(e) = v$, which contradicts $\{u,v\}$ being admissible.
\end{proof}
\begin{theorem}
If $\{u,v\}$ is admissible, and $e \in K$, $e \neq \{u,v\}, e_{1}', e_{2}'$ is paired with triangle $r$, then $\xiuv(e)$ is paired with $\xiuv(r)$.
\end{theorem}
\begin{proof}
Aiming for a contradiction, we assume that there exist an edge $e$ which is paired with a triangle $r$ but $\xiuv(e)$ is not paired with $\xiuv(r)$. In particular, we assume that $e$ is the greatest edge under $\prec$ that meets this condition. Therefore, for all triangle edges $e'$ where $e \prec e'$, $\xiuv(e')$ is paired with $\xiuv(\mint(e'))$. Hence, there necessarily exists a path in $K'_t$ from the edge corresponding to $e$ to the nodes corresponding to $\xiuv(\mint(e))$ and $\xiuv(\maxt(e))$, as contracting can only ``shorten'' the path. Because neither $\mint(e)$ nor $\maxt(e)$ are vanishing, this implies that $\xiuv(e)$ is paired with $\xiuv(\mint(e))$. But this means that $\maxt(e) \prec \mint(e)$, a contradiction. 
\end{proof}
\begin{corollary}
If $\{u,v\}$ is admissible, then $K'_p = \{(\xiuv(\sigma),\xiuv(\tau)) | (\sigma, \tau) \in K_p \land \xiuv(\sigma) \neq \xiuv(\tau)\}$.
\end{corollary}

\subsection{Expanding the Conditions}

These contraction conditions are expandable, particularly for the vertex/edge pairing. For example, it is somewhat easy to extend our conditions to permit $\{u,v\}$ to be paired with $u$. Expanding the triangle conditions is significantly more difficult. In addition, the notion of preserving the pairing generalizes to the case where neither mirror is paired with their shared vanishing simplex. If $\sigma$ is a vanishing simplex with mirrored facets $\tau_1,\tau_2$, where $\tau_1$ is paired with $\tau_1'$, $\tau_2$ is paired with $\tau_2'$, and $\sigma$ paired with $\sigma'$, then one may be interested in an operator such that $\xiuv(\tau_1)$ is paired with $\xiuv(\tau_1')$ and $\xiuv(\tau_2')$ is paired with $\xiuv(\sigma')$. It will be interesting to develop necessary and sufficient conditions for both of these problems. 
\section{Stable Contraction}
\label{sec:stable}

Let $\dgm_p(M_A)$ and $\dgm_p(M'_A)$ denote the $p$-dimensional persistence diagrams corresponding to the persistence modules $M_A$ and $M'_A$, where $M'_A$ is obtained by contracting a single edge $\{u,v\}$. If $d_B(\dgm_p(M_A),\dgm_p(M'_A)) \leq \epsilon$, then the contraction map $\xiuv$ is said to be $(p,\epsilon)$-stable. In this section we develop such a contraction operator. We will always assume that edge $\{u,v\}$ meets the link condition. In addition, we let $h$ and $h'$ denote height functions on $K$ and $K'$. 

Due to Theorem \ref{thm:bint}, it is sufficient to develop a contraction operator which bounds the interleaving distance between $M_A$ and $M_A'$. Note that any persistence module defined over some index set $A$ can be extended to a persistence module over $\mathbb{R}$ in a canonical way. Hence, we now refer to $M_\mathbb{R}$ and $M'_\mathbb{R}$ and will establish maps such that they are strongly interleaved. We let $\xi_{j_i}$ denote the restriction of $\xiuv$ to the subcomplex $K_{j_i}$. Now, we define the maps $\xi_{j_i}^* \; : H_p(K_{j_i}) \to H_p(K'_{j_i})$ in the following diagram.

\begin{center}
\begin{tikzcd}
\label{diag:first}
\cdots \arrow[r,"f_{j_0,j_1}"] & H_p(K_{j_1}) \arrow[r,"f_{j_1,j_2}"] \arrow[d,"\xi_{j_1}^*"]
 & H_p(K_{j_2}) \arrow[r,"f_{j_2,j_3}"] \arrow[d,"\xi_{j_2}^*"]
& H_p(K_{j_3}) \arrow[r,"f_{j_3,j_4}"] \arrow[d,"\xi_{j_3}^*"]
& \cdots \\
\cdots \arrow[r,"f'_{j_0,j_1}"] & H_k(K'_{j_1}) \arrow[r,"f'_{j_1,j_2}"]
& H_p(K'_{j_2}) \arrow[r,"f'_{j_2,j_3}"] 
& H_p(K'_{j_3}) \arrow[r,"f'_{j_3,j_4}"]
& \cdots
\end{tikzcd}
\end{center}

We extend $\xi_{j_i}$ to the canonical chain map $\xi_{j_i,\#} \; : \; C_p(K_{j_i}) \to C_p(K'_{j_i})$. This map induces a map between homology groups. Let $\gamma \in H_p(K_{j_i})$ and let $\sum_I \sigma_i$ be a representative cycle of $\gamma$. Then we define $\xi_{j_i}^* \; : \; H_p(K_{j_i}) \to H_p(K'_{j_i})$ where $\xi_{j_i}^*(\gamma)$ is determined by $\xi_{j_i,\#}(\sum_I(\sigma_i))$.

\begin{lemma}
If $\sum_I \sigma_i$, $\sigma_i \in K$ is a $p$-cycle in sublevel set $K_{j_i}$, then $\xi_{j_i,\#}(\sum_I \sigma_i)$ is also a $p$-cycle at sublevel set $K_{j_i}'$.
\end{lemma}
\begin{proof}
We aim to show that $\partial \xi_{j_i,\#}(\sum_I \sigma_i) = 0$. For every facet $\tau$ of some $\sigma_i$, there are necessarily an even number of $\sigma_i$ containing $\tau$ as a facet. It is sufficient to show that for nonvanishing $\tau$, $\xi_{j_i}(\tau)$ is a facet of an even number of summands of $\sum_I \xi_{j_i}(\sigma_i)$. If $\tau$ is nonlocal, then the result is immediate, as if $\sigma_i$ is incident to $\tau$, then $\xiuv(\sigma_i)$ is incident to $\xiuv(\tau) = \tau$. If $\tau$ is mirrored, then $\xiuv(\tau)$ is incident to the adjacent $\sigma_i$ which contain either $\tau$ or $m(\tau)$, the sum of which is clearly even. 
\end{proof}

\begin{theorem}
If $\sum_I \sigma_i$, $\sigma_i \in K$ is a boundary at sublevel set at $K_{j_i}$, then $\xi_{j_i,\#}(\sum_I \sigma_i)$ is also a boundary at sublevel set $K_{j_i}'$.
\label{cor:bdry}
\end{theorem}
\begin{proof}
Assume that $\sum_I \sigma_i = \partial(\Gamma)$ for some $\dim(\sigma_i)+1$ chain $\Gamma$. We show that $\xi_{j_i,\#}(\sum_I \sigma_i) = \partial \xi_{j_i,\#}(\Gamma)$.
First, consider $\sigma$ where $\sigma$ is the facet of an odd number of elements in $\Gamma$. If $\sigma$ is nonlocal, then $\xi_{j_i}(\sigma)$ is the facet of the same number of summands of $\Gamma$ as $\sigma$, so $\xi_{j_i}(\sigma)$ remains on the boundary. Similarly, if $\sigma$ is mirrored with odd parity and $m(\sigma)$ is incident to an even number of elements of $\Gamma$, then it is easy to see $\xi_{j_i}(\sigma)$ is the facet of an odd number of elements of $\xi_{j_i,\#}(\Gamma)$. If both $\sigma$ and $m(\sigma)$ are the facet of an odd number of elements of $\Gamma$, then $\xi_{j_i}(\sigma)$ is the facet of an even number. However, $\xi_{j_i}(\sigma) = \xi_{j_i}(m(\sigma))$, so they will cancel. This implies that every summand of $\xi_{j_i,\#}(\sum_I \sigma_i)$ that does not cancel maintains odd parity.

It remains to be seen that facets of an even number of elements of $\Gamma$ maintain even parity under $\xi_{j_i,\#}$. Let $\sigma$ be such a facet. If $\sigma$ is nonlocal, then this is trivially true. Similarly, if $\sigma$ is a mirror and $m(\sigma)$ is the facet of an even number of elements of $\Gamma$, then $\xi_{j_i}(\sigma)$ must be the facet of an even number of summands of $\xi_{j_i,\#}(\Gamma)$. 

Hence, $\xi_{j_i,\#}(\sum_I \sigma_i) = \partial \xi_{j_i,\#}(\Gamma)$, and the proof follows.

\end{proof}

The next theorem follows immediately from the previous two results.

\begin{theorem}
The maps $\xi_{j}^*$ are well defined.
\end{theorem}
For all $\epsilon > 0$, $\xi_{j}^*$ extends to a map $\xi_{j}^{j+\epsilon,*}$ by composing $f'_{j,j + \epsilon}$ with $\xi_{j}^*$. 

It is now necessary to define maps from the contracted persistence module to the original module. To do so, we will require $\{u,v\}$ to be $(p,\epsilon)$-admissible for some fixed $\epsilon$.
\begin{definition}
Edge $\{u,v\}$ is $(p,\epsilon)$-admissible if
\begin{enumerate}
    \item For each pair of $p$-mirrors $\sigma_1 \prec \sigma_2$ with shared vanishing cofacet $\tau$, $|h(\sigma_2) - h(\tau)| \leq \epsilon$ and,  
    \item if $p > 0$ then for each pair of $p-1$-mirrors $\sigma_1 \prec \sigma_2$ with shared vanishing cofacet $\tau$, $|h(\sigma_2) - h(\tau)| \leq \epsilon$. 
\end{enumerate}
\end{definition}
Intuitively, the first requirement  ensures commutativity when contracting $\{u,v\}$ destroys homological classes in some subcomplex. The second requirement does the same when contracting $\{u,v\}$ creates cycles in a subcomplex. Clearly, contraction cannot create a new $0$-dimensional class, so the requirement does not apply in this case.

For $(p,\epsilon)$-admissible edge $\{u,v\}$, we now define chain maps $\psi_{a}^b \; : \; K'_a \to C_p(K_b)$ $\forall a \in \mathbb{R}$, where $b = a + \epsilon$ and $C_p(K_b)$ is the group of $p$-chains over $K_b$. Of particular importance is determining the image under $\psi_a^b$ of those simplices $\sigma$ which are the image of mirrored simplices. If $\sigma = \xiuv(\tau) = \xiuv(m(\tau))$, $\tau \prec m(\tau)$, then we let $\psi_a^b(\sigma) = \tau + \sum_I \sigma_i$, where $\sigma_i$ are a subset of those vanishing $\dim(\sigma)$-simplices relative to $\{u,v\}$ which are incident to $\tau$. In particular, we only include those $\sigma_i$ where, for the shared mirrored facet $\eta$, $m(\eta) \prec \eta$. Similarly, if $\sigma$ is the image of an adjacent simplex which contains $\eta$ as a facet, $\eta$ a mirror where $m(\eta) \prec \eta$, then $\psi_a^b(\sigma) = \xiuv^{-1}(\sigma) + \tau$, where $\tau$ is the vanishing cofacet of $\eta$. If $\sigma$ is nonlocal, then $\psi_a^b(\sigma) = \sigma$. The map $\psi_{a}^b$ now extends linearly to a chain map $\psi_{a,\#}^b \; : \; C_p(K'_a) \to C_p(K_b)$.

\begin{theorem}
Let $\sum_I \sigma_i$ be a cycle in $K'_a$. Then $\psi_{a,\#}^b(\sum_I \sigma_i)$ is a cycle in $K_{a+\epsilon}$.
\label{thm:cycle}
\end{theorem}
\begin{proof}
Note that $\psi_{a,\#}^b(\sum_I \sigma_i)  = \sum_I \psi_a^b(\sigma_i)$. Since $\psi_a^b$ brings simplices to chains, we rewrite this sum as $\sum_{I'} \sigma_{i'}$. Let $F = \{\tau <_1 \sigma_{i'} \; | \; i' \in I'\}$. We aim to show that for each $\tau \in F$, $\tau$ is the facet of an even number of $\sigma_{i'}$. First, assume $\tau$ is the facet of nonlocal simplices. Then because $\psi_a^b$ is the identity on nonlocal simplices, $\tau$ is necessarily the facet of an even number of simplices in $\sum_{I'} \sigma_{i'}$. Now assume that $\tau$ is a mirrored simplex. Note that there is a unique $\dim(\tau)+1$ vanishing simplex containing $\tau$ as a facet. Call this simplex $\sigma$. Note that if the preimage of $\tau$ under $\psi_a^b$ is incident to an even number of  adjacent simplices, then $\tau$ is contained in an even number of mirrors, and is either contained in an even number of adjacent simplices containing $u$, an even number of adjacent simplices containing $v$, or an odd number of adjacent simplices and $\sigma$. In any case, $\tau$ is contained in an even number of $\sigma_{i'}$.  

Finally, we consider the case where $\tau$ is a vanishing simplex. Note that as $\tau$ is vanishing, it is only the facet of vanishing simplices. We consider two graphs, $G$ and $G'$. Let the vertex set of $G'$ correspond to the $\dim(\tau)$-simplices which contain $\xiuv(\tau)$ as a facet and appear as the facet of an element of $\sum \sigma_i$. Two vertices are connected if their corresponding $\dim(\tau)$-simplices are both the face of a single $\dim(\tau)+1$ simplex $\sigma_i$ which is an element of the cycle. Note that this graph is necessarily Eulerian, as every $\dim(\tau)$ simplex occurs as the facet of an even number of $\dim(\tau)+1$ simplices in the cycle. The graph $G$ is defined by the preimages of these simplices. We claim that $G$ is Eulerian. From $G'$, it is immediate that each vertex corresponding to a non-mirror is necessarily incident to an even number of edges. The parity for mirrors follows similarly. Hence, all vertices in $G'$ are incident to an even number of edges. Note that the vertex set of $G$ can be partitioned into $G_v$ and $G_u$, depending on if a vertex corresponds to a simplex containing $u$ or $v$. A walk of the Eulerian circuit in $G$ must then cross between $G_v$ and $G_u$ and even number of times. Each crossing corresponds to a vanishing simplex which contains $\tau$, which means that $\tau$ is a facet of an even number of simplices in $\psi_{a,\#}^b(\sum_I \sigma_i)$.
\end{proof}

The next theorem follows from essentially the same parity argument as Theorem \ref{thm:cycle}.
\begin{theorem}
Let $\sum_J \tau_j$ be a boundary in $K'_a$. Then $\psi_{a,\#}^b(\sum_J \tau_j)$ is a boundary in $K_{a + \epsilon}$. 
\end{theorem}
\begin{proof}
Let $\Gamma$ be a $\dim(\tau_j)+1$ chain such that $\sum_J \tau_j = \partial(\Gamma)$. First, consider the case where all $\tau_j$ are the images of nonlocal simplices. Let $\Gamma'$ denote the set of facets of elements of $\Gamma$ that are the facets of an even number of elements. If all $\sigma \in \Gamma'$ are the images of nonlocal simplices, then the proof follows immediately. Hence, assume that some $\sigma \in \Gamma'$ is the image of mirrored simplices. In particular, assume that $\xiuv(\sigma_u) = \xiuv(\sigma_v) = \sigma$. Since $\sigma$ is a facet of an even number of elements of $\Gamma$, then $\sigma_u$ and $\sigma_v$ are both either the facet of an odd number or even number of elements of adjacent simplices with elements in $\Gamma$. If they are incident to an odd number, then because of how $\psi_{a,\#}^b$ is defined, their shared vanishing cofacet is included in $\psi_{a,\#}^b(\Gamma)$, and $\sigma_u$ and $\sigma_v$ are the facets of an even number of simplices in $\psi_{a,\#}^b(\Gamma)$. 

In this case, all that remains to be shown is that if $\sigma$ is a vanishing simplex which is a facet of some element of $\psi_{a,\#}^b(\Gamma)$ and which is not an element of $\psi_{a,\#}^b(\tau_j)$ for some $\tau_j$ then $\sigma$ is the facet of an even number of elements of $\psi_{a,\#}^b(\Gamma)$. This argument is analogous to the Eulerian graph argument in the previous theorem. Consider all $\dim(\tau_j)$-simplices which contain $\xiuv(\sigma)$ as a facet as nodes. An edge is added between those $\dim(\tau_j)$-simplices which are both facets of a single $\dim(\tau_j)+1$-simplex. This graph is necessarily Eulerian, which implies that $\psi_{a,\#}^b$ induces an Eulerian graph, which implies that $\sigma$ is the facet of an even number of elements of $\psi_{a,\#}(\Gamma)$. 

Now, we consider the case where one $\tau_j$ is the image of mirrored simplices. In such a case, $\psi_a^b(\tau_j)$ may be $m(\tau_j) + \sum_I \gamma_i$ where $\gamma_i$ are vanishing simplices. In such a case, it is necessary to include the vanishing cofacet of $\tau_j$ with $\Gamma$. The remainder of the proof follows as in the first case. 
\end{proof}

Hence, $\psi_{a,\#}^b$ induces a map between homology groups which we denote $\psi_a^{b,*}$. We have already shown that $\xi_a^*$ extends to $\xi_a^{b,*}$, $b = a + \epsilon$, for all $a \in \mathbb{R}$. All that remains is to show that the maps $\xi_a^{b,*}$ and $\psi_a^{b,*}$ commute with the homomorphisms given by the persistence modules as in Equation \ref{eqn:homcommute}.

\begin{center}
\begin{equation}
\label{eqn:homcommute}
\begin{tikzcd}
H_p(K_{\alpha-\epsilon}) \arrow[r] \arrow[d] & H_p(K_{\alpha' + \epsilon}) & H_p(K_{\alpha + \epsilon}) \arrow[r] & H_p(K_{\alpha' + \epsilon}) \\
H_p(K'_\alpha) \arrow[r] & H_p(K'_{\alpha'}) \arrow[u] & H_p(K'_{\alpha}) \arrow[r] \arrow[u] & H_p(K'_{\alpha'}) \arrow[u] & \\
 H_p(K_\alpha) \arrow[r] & H_p(K_{\alpha'}) \arrow[d] & H_p(K_{\alpha}) \arrow[r] \arrow[d] & H_p(K_{\alpha'}) \arrow[d] & \\
H_p(K'_{\alpha-\epsilon}) \arrow[r] \arrow[u] & H_p(K'_{\alpha' + \epsilon}) & H_p(K'_{\alpha + \epsilon}) \arrow[r] & H_p(K'_{\alpha' + \epsilon})
\end{tikzcd}   
\end{equation}
\end{center}

\begin{theorem}
If $\{u,v\}$ is $(p,\epsilon)$-admissible, then the maps $\psi_{a}^{b,*}$, $\xi_a^{b,*}$, $f_{a,b}$ and ${f'}_{a,b}$ commute as in Equation \ref{eqn:homcommute}.
\end{theorem}
\begin{proof}
The commutation is immediate, as the chain maps commute except for on mirrored simplices, where if $v < \tau <_1 \sigma$ $\psi_a^{b,*}$ may bring $\xiuv(\tau)$ to $\partial \sigma \setminus \tau$. But then the cycles differ by a boundary, as the link condition requires that each pair of mirrors share a vanishing cofacet, and $b$ is chosen such that the shared cofacet is guaranteed to be in $K_b$.  Hence, the induced homology maps commute. 
\end{proof}

\begin{corollary}
If $\{u,v\}$ satisfies the link condition and is $(p,\epsilon)$-admissible, then $\xiuv$ is $(p,\epsilon)$-stable. 
\end{corollary}

\subsection{Multiple Contractions}
The existence of $(p,\epsilon)$-stable contraction gives rise to the question of bounding interleaving distance across multiple contractions. Consider $(p,\epsilon)$-admissible edges, $\{u_1,v_1\}$,  $\{u_2,v_2\}$, $\{u_3,v_3\}$, $\ldots$, $\{u_n,v_n\}$. Specifically, $\{u_2,v_2\}$ is $(p,\epsilon)$-admissible in the complex $\xi_{\{u_1,v_1\}}(K)$, $\{u_3,v_3\}$ is $(p,\epsilon)$-admissible in the complex $\xi_{\{u_2,v_2\}} \circ \xi_{\{u_1,v_1\}}(K)$, and so on. Contracting these edges sequentially gives a sequence of $p$-dimensional persistence modules $M_0, M_1, \ldots, M_n$. Naively applying the triangle inequality implies that $d_b(\dgm_p(M_0), \dgm_p (M_n)) \leq n\epsilon$. We aim to find conditions under which $d_b(\dgm_p(M_0),\dgm_p(M_n)) \leq \epsilon$. In this section, we will use $K^{(i)}$ to denote the complex that results from contracting the first $i$ edges, and will use $f^{(i)}$ to refer to the inclusion homomorphisms in $M_i$. 

For each $(p,\epsilon)$-admissible edge $e \in K$, we define a $p$-\textit{window} $W_p(e) \subset \mathbb{R}$.  If $0 < p < \dim(K)$, let $r$ be the $p-1$ mirror relative to $e$ with minimal height value, and $s$ be the $p+1$ simplex with maximal height value. Then $W_p(e) = [f(r), f(s)]$. If $p=0$, then $W_p(e) = [f(r),f(e)]$. Let $z$ denote the maximum vanishing $p$-simplex relative to $\{u,v\}$. If $p = \dim(K)$, then $W_p(e) = [f(r), f(z)]$. 

\begin{definition}
The sequence of $(p,\epsilon)$-admissible edges $\{u_1,v_1\} \in K$, $\{u_2,v_2\} \in K'$, $\ldots$, $\{u_n,v_n\} \in K^{(n-1)}$ is $(p,\epsilon)$-compatible if the intervals $W_p(\{u_1,v_1\})$, \ldots, $W_p(\{u_n,v_n\})$ are disjoint.
\end{definition}
This definition permits us to bound the interleaving distance between the initial and final persistence modules.
Downward arrows are now given by the composition of contraction maps. Each given $\{u_i,v_i\}$ has a different collection of $\psi_a^b$, so we use the notation $\psi_{i,a}^b \; : \; H_p(K_a^{(i)}) \to H_p(K_b^{(i-1)})$ to distinguish them. Note that if $W_p(\{u_i,v_i\}) = [r,s]$, then $H_p(K_j^{(i-1)}) \cong H_p(K_j^{(i)})$ provided $j \not\in [r,s)$. It is easy to see that careful composition of $\psi_{i,a}^b$ with these isomorphisms defines a new collection of upward maps $\Psi_a^b \; : \; H(K_a^{(n)}) \to H(K_b)$, $b = a + \epsilon$. Similarly, through composition of contraction and inclusion maps, we get $\Xi_a^b \; : H_p(K_a) \to H_p(K_b^{(n)})$. It is clear that both collections of maps inherit the boundary to boundary and cycle to cycle properties from the single-contraction case. Hence, the following theorem follows immediately. 

\begin{center}
\begin{equation}
\label{eqn:homcommutebig}
\begin{tikzcd}
H_p(K_{\alpha-\epsilon}) \arrow[r] \arrow[d] & H_p(K_{\alpha' + \epsilon}) & H_p(K_{\alpha + \epsilon}) \arrow[r] & H_p(K_{\alpha' + \epsilon}) \\
H_p(K^{(n)}_\alpha) \arrow[r] & H_p(K^{(n)}_{\alpha'}) \arrow[u] & H_p(K^{(n)}_{\alpha}) \arrow[r] \arrow[u] & H_p(K^{(n)}_{\alpha'}) \arrow[u] & \\
 H_p(K_\alpha) \arrow[r] & H_p(K_{\alpha'}) \arrow[d] & H_p(K_{\alpha}) \arrow[r] \arrow[d] & H_p(K_{\alpha'}) \arrow[d] & \\
H_p(K^{(n)}_{\alpha-\epsilon}) \arrow[r] \arrow[u] & H_p(K^{(n)}_{\alpha' + \epsilon}) & H_p(K^{(n)}_{\alpha + \epsilon}) \arrow[r] & H_p(K^{(n)}_{\alpha' + \epsilon})
\end{tikzcd}   
\end{equation}
\end{center}

\begin{theorem}
Let $\{u_1,v_1\} \in K$, $\{u_2,v_2\} \in K'$, $\ldots$, $\{u_n,v_n\} \in K^{(n-1)}$ be $(p,\epsilon)$-compatible edges. The maps $\Psi_a^b$, $\Xi_a^b$, $f$, and $f^{(n)}$ commute as in Equation \ref{eqn:homcommutebig}. 
\end{theorem}
\section{Experiments}
\label{sec:exps}

In this section, we implement our contraction operator and demonstrate its utility on some manifolds. Note that the notion of $(p,\epsilon)$-compatibility lends itself to an elementary scheduling problem. At each stage, we aim to contract as many edges as possible while ensuring that the bottleneck distance between the previous and resulting persistence diagrams remains $\leq \epsilon$. Hence, by the triangle inequality, the bottleneck distance between persistence diagrams for the initial and contracted complexes is $\leq m\epsilon$, where $m$ is the number of stages. In this section, we will contract sets of edges which are $(1,0.5)$-compatible and $(1,5.0)$-compatible at each stage. 

\begin{figure}\centering
  \frame{\includegraphics[scale=.175]{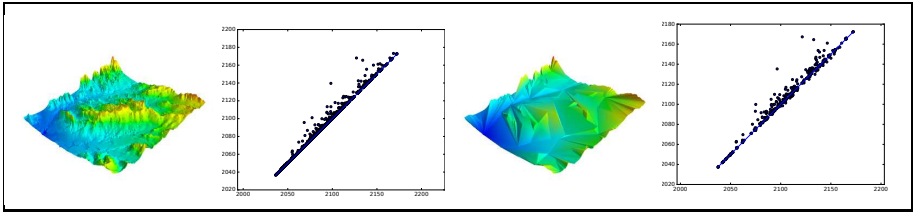}}
  \caption{A triangulated terrain (far left), the $1$-dimensional persistence diagram corresponding to the terrain (center left), the complex after approximately 450,000 edge contractions resulting in a 99\% reduction in the number of simplices (center right), the $1$-dimensional persistence diagram corresponding to the contracted complex (right). The bottleneck distance between the two persistence diagrams is 6.540.}
  \label{fig:contEx}
\end{figure}
\begin{figure}\centering
 \frame{\includegraphics[scale=.28]{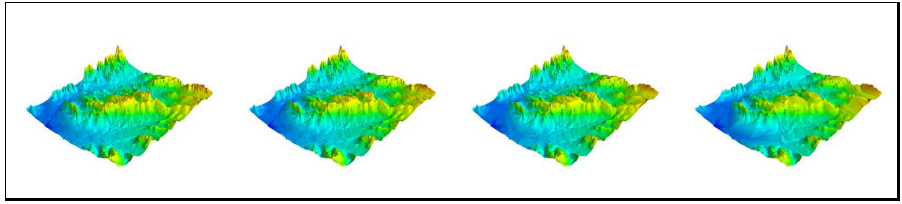}}
  \caption{The Los Alamos dataset after, from left to right, 100,000, 200,000, 300,000, and 400,000 $(1,5.0)$-admissible edge contractions.}
  \label{fig:contGrad}
\end{figure}
\begin{figure}\centering
  \frame{\includegraphics[scale=.28]{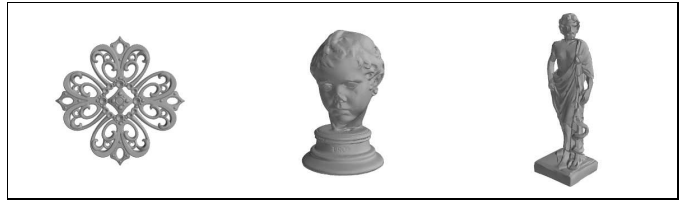}}
  \caption{The non terrain models: Filigree (genus 65), Eros (genus 0), and Statue (genus 4).}
  \label{fig:models}
\end{figure}

\begin{table*}
\small
\centering
 \begin{tabular}{||l | c | c | c | c | c | r||} 
 \hline
 Dataset & Init. Simplices & Contractions & Its. & Rem. Simps & \% Red. & $d_b$\\ [0.5ex] 
 \hline\hline
 Columbus & 2,728,353 & 450,835 & 6,092 & 23,343 & 99.14 & 0.601\\ 
 \hline
 Los Alamos & 2,728,353 & 449,181 & 7,840 & 33,267 & 98.78 & 1.370\\
 \hline
 Minneapolis & 10,940,401 & 1,818,560 & 66,627 & 29,041 & 99.73 & 1.106 \\
 \hline
 Aspen & 10,940,401 & 1,806,640 & 6,231 & 100,561 & 99.08 & 1.115\\
 \hline
 Filigree & 3,086,440 & 438,925 & 85,116 & 452,890 & 85.32 & 0.916\\
 \hline
 Eros & 2,859,566 & 475,401 & 157,636 & 7,160 & 99.75 & 0.489\\
 \hline
 Statue & 14,994 & 2,470 & 945 & 174 & 98.84 & 0.074\\
 \hline
\end{tabular}
\vspace{11pt}
\caption{Contraction data for seven datasets when contracting $(1,0.5)$-compatible edges. Iterations corresponds to the number of sets of compatible edges that were contracted. Despite many iterations, note that the bottleneck distance between the $1$-dimensional persistence diagrams remains comparatively close to $\epsilon = 0.5$.}
\label{tabl:contDataHalf}
\end{table*}

\begin{table*}
\small
\centering
 \begin{tabular}{||l | c | c | c | c | c | r||} 
 \hline
 Dataset & Init. Simplices & Contractions & Its. & Rem. Simplices & \% Red. & $d_b$\\ [0.5ex] 
 \hline\hline
 Columbus & 2,728,353 & 452,914 & 6,123 & 10,869 & 99.60 & 2.097\\ 
 \hline
 Los Alamos & 2,728,353 & 452,710 & 7,350 & 12,093 & 99.56 & 6.540\\
 \hline
 Minneapolis & 10,940,401 & 1,819,786 & 66,981 & 21,685 & 99.80 & 4.832 \\
 \hline
 Aspen & 10,940,401 & 1,818,593 & 5,664 & 28,843 & 99.74 & 8.570\\
 \hline
 Filigree & 3,086,440 & 466,395 & 90,009 & 288,070 & 90.67 & 9.051\\
 \hline
 Eros & 2,859,566 & 476,573 & 158,426 & 128 & 99.99 & 4.296\\
 \hline
 Statue & 14,994 & 2,470 & 945 & 174 & 98.84 & 0.074\\
 \hline
\end{tabular}
\vspace{11pt}
\caption{Contraction data for seven datasets when contracting $(1,5.0)$-compatible edges. As in the $(1,0.5)$ case, the bottleneck distance between the $1$-dimensional persistence diagrams remains comparatively close to $\epsilon = 5.0$.}
\label{tabl:contData5}
\end{table*}

We consider four terrains, obtained from the National Elevation Dataset, and $3$ models, obtained from the Aim@Shape repository. The terrains are from near Aspen, Colorado; Minneapolis, Minnesota; Los Alamos, New Mexico; and Columbus, Ohio. A height function is defined over the models' vertex set by a surface curvature approximation. For each manifold, we find a maximal set of $(1,\epsilon)$-compatible edges, contract them all, and repeat the process until there are no remaining $(1,\epsilon)$-admissible edges. An example of the contraction process and persistence diagrams can be seen in Figures \ref{fig:contEx} and \ref{fig:contGrad}. In Tables \ref{tabl:contDataHalf} and \ref{tabl:contData5} we show the resulting data for $\epsilon = 0.5$ and $\epsilon = 5.0$, respectively.

In general, the bottleneck distance is much less than $m\epsilon$. This raises an issue of the tightness of our multiple contraction scheme, which we leave to future work. 
\section{Conclusion}
We conclude with a short discussion on directions for future research. It appears that the conditions for $(p,\epsilon)$-compatible edges are very conservative, and could possibly be expanded to permit further contraction. In addition, there are a variety of other operations that could be applied to simplicial complexes which may be persistence-sensitive. Based on the theory of strong homotopy~\cite{Barmak}, Boissonat \textit{et al.} have recently invented a quick way to compute persistent homology by simplifying the input complex with strong collapses \cite{Boisson}. A natural problem is to develop an elementary collapse operator that controls the perturbation in the persistence diagrams. Other such operations worthy of investigation include vertex removal, or arbitrary vertex identification in a CW-complex.

\subsection{Acknowledgments}

The authors would like to thank the National Elevation Dataset for their terrain data, the Aim@Shape repository for the models, and the Hera project for their bottleneck distance code \cite{Kerber}. In addition, the authors are grateful for the comments of the anonymous reviewers. This work was supported by NSF grants CCF-1740761, DMS-1547357 and CCF-1839252.

\bibliographystyle{abbrv}
\bibliography{refs}
\end{document}